\documentclass[11pt]{article}
  
\usepackage{cite}
\usepackage{graphicx}
\usepackage{amsmath}
\usepackage{amsthm}
\usepackage{yhmath} 

\usepackage{amsfonts} 
\usepackage{amssymb}
\usepackage{fullpage}
\usepackage{color}
\usepackage{latexsym}

\usepackage{algorithm}
\usepackage[noend]{algorithmic}
\usepackage{multirow} 

\newcommand{\old}[1]{{}}

\newtheorem{theorem}{Theorem}[section]
\newtheorem{corollary}[theorem]{Corollary}
\newtheorem{lemma}[theorem]{Lemma}

\newtheorem{obs}[theorem]{Observation}

\renewcommand{\r}{{{\rho}}}

\renewcommand{\r}{{{\rho}}}
\newcommand{\mwst}{{minimum-weight sink tree }}

\title{Minimizing Total Interference in Asymmetric Sensor Networks
\footnote{This work was partially supported by Grant 2016116 from the United States -- Israel Binational Science Foundation.}}
 
\author{ A. Karim Abu-Affash\thanks{Software Engineering Department, Shamoon College of Engineering, Beer-Sheva 84100, Israel, {\tt abuaa1@sce.ac.il}.}
\and 
Paz Carmi\thanks{Department of Computer Science, Ben-Gurion University, Beer-Sheva 84105, Israel, {\tt carmip@cs.bgu.ac.il}.}
\and 
Matthew J. Katz\thanks{Department of Computer Science, Ben-Gurion University, Beer-Sheva 84105, Israel, {\tt matya@cs.bgu.ac.il}.}
} 

\begin{document}

\maketitle

\begin{abstract}
The problem of computing a connected network with minimum interference is a fundamental problem in wireless sensor networks. Several models of interference have been studied in the literature. The most common model is the receiver-centric, in which the interference of a node $p$ is defined as the number of other nodes whose transmission range covers $p$. In this paper, we study the problem of assigning a transmission range to each sensor, such that the resulting network is strongly connected and the total interference of the network is minimized. 

For the one-dimensional case, we show how to solve the problem optimally in $O(n^3)$ time. For the two-dimensional case, we show that the problem is NP-complete and give a polynomial-time 2-approximation algorithm for the problem. 
\end{abstract}


\section{Introduction}

{\em Wireless sensor networks} have received significant attention in the last two decades due to their potential civilian and military applications~\cite{Estrin, Hubaux}. A wireless sensor network consists of numerous devices that are equipped with processing, memory, and wireless communication capabilities.
This kind of network has no pre-installed infrastructure, rather all communication is supported by multi-hop transmissions, where intermediate nodes relay packets between communicating parties.
Since each sensor has a limited battery and each transmission decreases the battery charge, energy consumption is a critical issue in wireless sensor networks. One well way to conserve energy is to minimize the interference of the network. 
High interference increases the probability of packet collisions and therefore packet retransmission, which can significantly affect the effectiveness and the lifetime of the network.

In wireless networks design, the nodes are modeled as a set of points in the plane and each node $u$ is assigned a transmission range $\r(u)$. A node $v$ can receive the signal transmitted by a node $u$ if and only if $|uv| \le \r(u)$, where $|uv|$ is the Euclidean distance between $u$ and $v$. 
There are two common ways to model the induced communication graph: symmetric and asymmetric. 
In the symmetric model~\cite{Halldorsson08,Lam10,Moscibroda05,Tan11,Rickenbach05}, there is a non-directed edge between points $p$ and $q$ if and only if $|pq| \le \min \{\r(p), \r(q) \}$, and in the asymmetric model~\cite{Agrawal13,Bilo08,Brise15,Rickenbach09}, there is a directed edge from node $p$ to node $q$ if and only if $|pq| \le \r(p)$. 

Several models of interference have been studied in the literature.
The model of Burkhart et al.~\cite{Burkhart} measures the number of nodes affected by the communication on a single edge. This was also studied by Moaveni-Nejad and Li~\cite{Moaveni-nejad}, who further introduced the {\em sender-centric} model that measures the number of receiving nodes affected by the communication from a single sender. They showed that both problems can be solved optimally in polynomial-time, in the symmetric and asymmetric models.

Von Rickenbach et al.~\cite{Rickenbach09} argued that the sender-centric model of interference is not realistic, because the interference is actually felt by the receiver. Further, they argued that the model was overly sensitive to the addition of single nodes. Instead, they formulated the {\em receiver-centric} model; minimizing the maximum interference received at a node, where the interference received at a node $p$ is the number of nodes $q$ ($q \ne p$) such that $|pq|\le \r(q)$.
In the one-dimensional case, i.e., when the points are located on a line, Von Rickenbach et al.~\cite{Rickenbach09} showed that one can always construct a network with $O(\sqrt{n})$ interference, in the symmetric model. Moreover, they showed that there exists an instance that requires $\Omega(\sqrt{n})$ interference, and gave an $O(n^{1/4})$ approximation algorithm for this case. Tan et al.~\cite{Tan11} proved that the optimal network has some interesting properties and, based on these properties, one can compute a network with minimum interference in sub-exponential time.
Brise et al.~\cite{Brise15} extend this result for the asymmetric model.

In the two-dimensional case, the problem has been shown to be NP-hard~\cite{Brise15,Buchin08}. Halld{\'o}rsson and Tokuyama~\cite{Halldorsson08} showed how to construct a network with $O(\sqrt{n})$ interference for the symmetric model, extending the result of~\cite{Rickenbach09}. For the asymmetric model, Fussen et al.~\cite{Fussen05} showed that one can always construct a network with $O(\log{n})$ interference and showed that their algorithm is asymptotically optimal.

Another model of interference is to minimize the total interference of the network. That is, given a set $P$ of points in the plane, the goal is to assign ranges to the points to obtain a connected communication graph in which the total interference is minimized. For the symmetric model, Moscibroda and Wattenhofer~\cite{Moscibroda05} studied the problem in general metric graphs. They showed that the problem is NP-complete and cannot be approximated within $O(\log{n})$, and they gave an $O(\log{n})$ approximation algorithm for this problem. Lam et al.~\cite{Lam10} proved that the problem is NP-complete for points in the plane. For the one-dimensional case, Tan et al.~\cite{Tan11} showed how to compute an optimal network in $O(n^4)$ time.

For the asymmetric model, Bil{\`o} and Proietti~\cite{Bilo08} studied the problem of minimizing the total interference in general metric graphs. They gave a logarithmic approximation algorithm for the problem by reducing it to the power assignment problem in general graphs. Agrawal and Das~\cite{Agrawal13} studied the problem in two-dimensions. They gave two heuristics with experimental results.

\paragraph{\textbf{Our results.}} In this paper, we consider the problem of minimizing the total interference in the asymmetric model. We first give an $O(n^3)$-time algorithm that solves the problem optimally, when the points are located on a line. Then, we show that the problem is NP-complete for points in the plane and give a 2-approximation algorithm for the problem. 
Our approach in the solution of the one-dimensional problem is somewhat unconventional and involves assigning to each point both a left range and a right range, see Section~\ref{sec:1d}. We note that the approach of Tan et al.~\cite{Tan11} (who considered the problem in the symmetric model) would yield a significantly worse time bound in our case, and that it would be interesting to check whether our approach can be applied also in their setting to obtain an improved running time.


\section{Network Model and Problem Definition}

Let $P$ be a set of sensors in the plane. Each sensor $p \in P$ is assigned a transmission range $\r(p)$. 
A sensor $q$ can receive a signal from a sensor $p$ if and only if $q$ lies in the transmission area of $p$.
We consider $P$ as a set of points in the plane and the range assignment to the sensors as a function $\r:P \to \mathbb{R}^+$. The communication graph induced by $P$ and $\r$ is a \emph{directed} graph $G_\r=(P,E)$, such that $E=\{(p,q):|pq| \leq \r(p) \}$, where $|pq|$ is the Euclidean distance between $p$ and $q$. $G_\r$ is \emph{strongly connected} if, for every two points $p,q\in P$, there exists a directed path from $p$ to $q$ in $G_\r$. A range assignment $\r$ is called \emph{valid} if the induced graph $G_\r$ is strongly connected.

Given a communication graph $G_\r=(P,E)$, in the \emph{receiver-centric} interference model, the interference of a point $p$, denoted by $RI(p)$, is defined as the number of points in $P \setminus \{p\}$ whose transmission range covers $p$, i.e., $RI(p)=|\{q \in P \setminus \{p\} :|pq| \leq \r(q) \}|$.
In the \emph{sender-centric} interference model, the interference of a point $p$, denoted by $SI(p)$, is defined as the number of points in $P \setminus \{p\}$ within the transmission range of $p$, i.e., $SI(p)=|\{q \in P \setminus \{p\} :|pq| \leq \r(p) \}|$. The \emph{total interference} of $G_\r$, denoted by $I(G_\r)$, is defined as 
$$I(G_\r) = \sum_{p\in P}RI(p) = \sum_{p\in P}SI(p) \ .$$

In the \emph{Minimum Total Interference} ($MTIP$) problem, we are given a set $P$ of points in the plane and the goal is to find a range assignment $\r$ to the points of $P$, such that the graph $G_\r$ induced by $P$ and $\r$ is strongly connected and $I(G_\r)$ is minimized.


\section{\emph{MTIP} in 1D}\label{sec:1d}

In this section, we present an exact algorithm that solves \emph{MTIP} in $O(n^3)$ time in 1D.
Let $P=\{p_1,p_2,\dots,p_n\}$ be a set of $n$ points located on a horizontal line. 
For simplicity, we assume that for every $i<j$, $p_i$ is to the left of $p_j$.

For each $1 \le i \le j \le n$, let $P_{[i,j]} \subseteq P$ be the set $\{p_i,p_{i+1},\dots ,p_j\}$. 
A \emph{sink} tree $T_{[i,j]}^x$ of $P_{[i,j]}$ rooted at $p_x$, where $x \in \{i,j\}$, is a directed tree that contains a directed path from each point $p \in P_{[i,j]}\setminus \{p_x\}$ to $p_x$.
Let $\r_{[i,j]}^x$ be a range assignment to the points in $P_{[i,j]}$, such that the graph induced by $\r_{[i,j]}^x$ contains a sink tree $T_{[i,j]}^x$ of $P_{[i,j]}$ rooted at $p_x$, and let $I(T_{[i,j]}^x)$ denote the total interference of $T_{[i,j]}^x$. 
We say that $\r_{[i,j]}^x$ is an optimal range assignment to the points in $P_{[i,j]}$ if and only if $I(T_{[i,j]}^x)$ is minimized.

Let $G=(P,E)$ be the complete directed graph on $P$. For each directed edge $(p_i,p_j) \in E$, we assign a weight $w(p_i,p_j) =|\{ p_k\in P\setminus \{p_i\}: |p_ip_k| \le |p_ip_j| \}|.$
Let $G_{[i,j]}$ be the subgraph of $G$ induced by $P_{[i,j]}$.
Let $T_{[i,j]}^x$ be a sink tree rooted at $p_x$ in $G_{[i,j]}$ and let $w(T_{[i,j]}^x) = \sum_{(p,q) \in T_{[i,j]}^x} w(p,q)$ denote its weight.
Let $\r_{[i,j]}^x$ be a range assignment to the points of $P_{[i,j]}$, such that its induced graph contains a sink tree $T_{[i,j]}^x$ rooted at $p_x$.
Observe that if $\r_{[i,j]}^x$ is an optimal range assignment to the points of $P_{[i,j]}$, then its induced graph contains a unique sink tree $T_{[i,j]}^x$.
\begin{lemma} \label{lemma:optsink}
$\r_{[i,j]}^x$ is an optimal range assignment to the points of $P_{[i,j]}$ if and only if $T_{[i,j]}^x$ is a minimum-weight sink tree rooted at $p_x$ in $G_{[i,j]}$.
\end{lemma}
\begin{proof}
$T_{[i,j]}^x$ is a sink tree in $G_{[i,j]}$ and, by the way we assigned weights to the edges of $G$, $I(T_{[i,j]}^x)$ is equal to the weight of $T_{[i,j]}^x$ in $G_{[i,j]}$, i.e., $I(T_{[i,j]}^x) = w(T_{[i,j]}^x)$. On the other hand, let $T$ be a sink tree rooted at $p_x$ in $G_{[i,j]}$ of weight $w(T)$. We define a range assignment $\r'$ to the points of $P_{[i,j]}$ as follows: Set $\r'(p_x)=0$, and for each $(p_k,p_l) \in T$, set $\r'(p_k)=|p_kp_l|$. Consider the graph induced by $\r'$. Clearly, this graph contains $T$ and $I(T) = w(T)$. This implies that $\r_{[i,j]}^x$ is an optimal range assignment to the points of $P_{[i,j]}$ if and only if $T_{[i,j]}^x$ is a minimum-weight sink tree rooted at $p_x$ in $G_{[i,j]}$.
\end{proof}

By Lemma~\ref{lemma:optsink}, to compute an optimal range assignment $\r_{[i,j]}^x$, it is sufficient to compute a minimum-weight sink tree rooted at $p_x$ in $G_{[i,j]}$. In Section~\ref{sec:allSinks}, we show how to compute a minimum-weight sink tree $T_{[i,j]}^x$ in $G_{[i,j]}$, for every $1 \le i \le j \le n$ and for every $x \in \{i,j\}$. Then, in Section~\ref{sec:MTIP}, we use these trees to devise a dynamic programming algorithm that solves \emph{MTIP} in $O(n^3)$ time. 

\subsection{Computing all sink trees} \label{sec:allSinks}
In this section, we show how to compute a \mwst rooted at $p_i$ and a \mwst rooted at $p_j$ in $G_{[i,j]}$, for every $1 \le i \le j \le n$, in $O(n^3)$ time.
The following observation follows from the way we constructed $G$.
%
\begin{obs} \label{obs:monotonic}
Let  $p_i$, $p_j$, and $p_k$ be three points of $P$. If $|p_i p_j| < |p_i p_k|$, then 
$w(p_i,p_j) < w(p_i,p_k)$.
\end{obs}

Let $T_{[i,j]}^x$ be a \mwst rooted at $p_x$ in $G_{[i,j]}$, where $x \in \{i,j\}$,
 and let $OPT_{[i,j]}^x$ denote its weight. 
The following lemma reveals the special structure of $T_{[i,j]}^x$.

\begin{lemma} \label{lemma:noCrossing}
For every two distinct edges $(p_k,p_l)$ and $(p_{k'},p_{l'})$ in $T_{[i,j]}^x$, 
\begin{enumerate}
	\item[$(i)$] if $k < {k'} < l$, then  $k \leq {l'} < l$; and   
  \item[$(ii)$] if $k < {l'} < l$, then  $k < {k'} < l$.
\end{enumerate}  
\end{lemma} 

\begin{proof}
\textcolor{white}{zzzz}
\begin{enumerate}
	\item[$(i)$] Assume towards a contradiction that there exist two edges $(p_k,p_l)$ and $(p_{k'},p_{l'})$, such that $k < {k'} < l$, and ${l'} < k $ or ${l'} \geq l$; see Figure~\ref{fig:noCrossing}.   
Notice that there is no path from $p_l$ to $p_{k'}$ in $T_{[i,j]}^x$, 
otherwise, by replacing $(p_k,p_l)$ by $(p_k,p_k')$ in $T_{[i,j]}^x$, we obtain a sink tree rooted at $p_x$ in $G_{[i,j]}$ of weight less than $OPT_{[i,j]}^x$, which contradicts the minimality of $T_{[i,j]}^x$.
Thus, the path from $p_k$ to $p_x$ in $T_{[i,j]}^x$ does not contain $p_{k'}$.
We distinguish between two cases. \\
\textbf{Case~1:}  ${l'} < k$; see Figure~\ref{fig:noCrossing}(a). By Observation~\ref{obs:monotonic}, $w(p_{k'},p_k) < w(p_{k'},p_{l'})$. Therefore, by replacing $(p_{k'},p_{l'})$ by $(p_{k'},p_k)$ in $T_{[i,j]}^x$, we obtain a sink tree rooted at $p_x$ in $G_{[i,j]}$ of weight less than $OPT_{[i,j]}^x$, which contradicts the minimality of $T_{[i,j]}^x$. \\
\textbf{Case~2:} ${l'} \geq l$. If ${l'} = l$, then, by Observation~\ref{obs:monotonic}, $w(p_k,p_{k'}) < w(p_{k},p_{l})$; see Figure~\ref{fig:noCrossing}(b). Therefore, by replacing $(p_{k},p_{l})$ by $(p_k,p_{k'})$ in $T_{[i,j]}^x$, we obtain a sink tree rooted at $p_x$ in $G_{[i,j]}$ of weight less than $OPT_{[i,j]}^x$, which contradicts the minimality of $T_{[i,j]}^x$. And, if ${l'} > l$, then, by Observation~\ref{obs:monotonic}, $w(p_{k'},p_l) < w(p_{k'},p_{l'})$; see Figure~\ref{fig:noCrossing}(c). Therefore, by replacing $(p_{k'},p_{l'})$ by $(p_{k'},p_l)$ in $T_{[i,j]}^x$, we obtain a sink tree rooted at $p_x$ in $G_{[i,j]}$ of weight less than $OPT_{[i,j]}^x$, which contradicts the minimality of $T_{[i,j]}^x$. 
\begin{figure}[ht]
	\centering
			\includegraphics[width=0.9\textwidth]{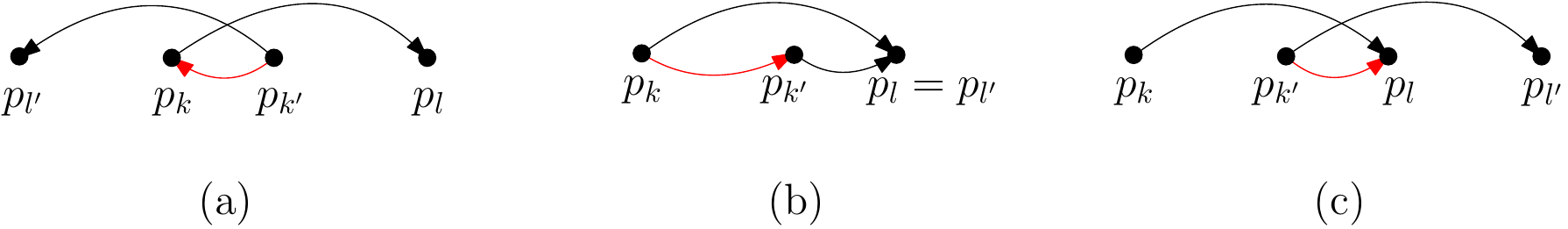}
	\caption{Illustration of the proof of Lemma~\ref{lemma:noCrossing}.} 
	\label{fig:noCrossing}
\end{figure}
%
	\item[$(ii)$] Assume towards a contradiction that there exist two edges $(p_k,p_l)$ and $(p_{k'},p_{l'})$, such that $k < {l'} < l$, and ${k'} \le k $ or ${k'} \ge l$; see Figure~\ref{fig:noCrossing1}. We distinguish between two cases. \\ 
\textbf{Case~1:} ${k'} < k$; see Figure~\ref{fig:noCrossing1}(a). In this case, there is no path from $p_{l'}$ to $p_{k}$ in $T_{[i,j]}^x$, otherwise, by replacing $(p_{k'},p_{l'})$ by $(p_{k'},p_k)$ in $T_{[i,j]}^x$, we obtain a sink tree rooted at $p_x$ in $G_{[i,j]}$ of weight less than $OPT_{[i,j]}^x$, which contradicts the minimality of $T_{[i,j]}^x$. 
Thus, the path from $p_{k'}$ to $p_x$ in $T_{[i,j]}^x$ does not contain $p_{k}$. 
Moreover, by Observation~\ref{obs:monotonic}, $w(p_{k},p_{l'}) < w(p_{k},p_{l})$. Therefore, by replacing $(p_{k},p_{l})$ by $w(p_{k},p_{l'})$ in $T_{[i,j]}^x$, we obtain a sink tree rooted at $p_x$ in $G_{[i,j]}$ of weight less than $OPT_{[i,j]}^x$, which contradicts the minimality of $T_{[i,j]}^x$. \\
\textbf{Case~2:} ${k'} > l$; see Figure~\ref{fig:noCrossing1}(b).
Notice that, since $T_{[i,j]}^x$ is a sink tree, either no path from $p_{l'}$ to $p_k$ or 
no path from $p_l$ to $p_{k'}$ exists in $T_{[i,j]}^x$ (otherwise, $T_{[i,j]}^x$ contains a cycle). 
Assume, w.l.o.g., that there is no path from $p_l$ to $p_{k'}$ in $T_{[i,j]}^x$. Thus, the path from $k$ to $x$ in $T_{[i,j]}^x$ does not contain the point $p_{k'}$.
By Observation~\ref{obs:monotonic}, $w(p_{k'},p_{l}) < w(p_{k'},p_{l'})$. Therefore, by replacing $(p_{k'},p_{l'})$ by $(p_{k'},p_{l})$ in $T_{[i,j]}^x$, we obtain a sink tree rooted at $p_x$ in $G_{[i,j]}$ of weight less than $OPT_{[i,j]}^x$, which contradicts the minimality of $T_{[i,j]}^x$. 
\end{enumerate}
\begin{figure}[ht]
	\centering
			\includegraphics[width=0.77\textwidth]{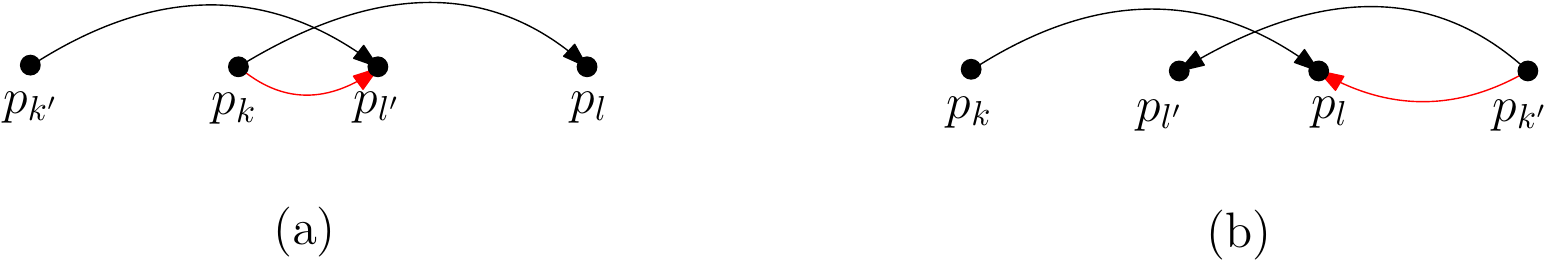}
	\caption{Illustration of the proof of Lemma~\ref{lemma:noCrossing}.} 
	\label{fig:noCrossing1}
\end{figure}
\end{proof}

Consider a \mwst $T_{[i,j]}^i$ of weight $OPT_{[i,j]}^i$ rooted at $p_i$ in $G_{[i,j]}$.
Let $(p_j,p_k)$ be the edge connecting $p_j$ to $p_k$ in $T_{[i,j]}^i$. Thus, $(p_j,p_k)$ partitions $T_{[i,j]}^i$ into two sub-trees $T^i$ rooted at $p_i$ and $T^j$ rooted at $p_j$; see Figure~\ref{fig:opt}. 
By Lemma~\ref{lemma:noCrossing}, $T^i$ contains the points of $P_{[i,k]}$ and $T^j$ contains the points of $P_{[k+1,j]}$. Moreover, since $T_{[i,j]}^i$ is a \mwst rooted at $p_i$ in $G_{[i,j]}$, $T^i$ is a \mwst rooted at $p_i$ in $G_{[i,k]}$ and $T^j$ is a \mwst rooted at $p_j$ in $G_{[k+1,j]}$. Therefore, $OPT_{[i,j]}^i = OPT_{[i,k]}^i + w(p_j,p_k) + OPT_{[k+1,j]}^j$. Similarly, $OPT_{[i,j]}^j = OPT_{[i,k-1]}^i + w(p_i,p_k) + OPT_{[k,j]}^j$.
\begin{figure}[ht]
	\centering
			\includegraphics[width=0.82\textwidth]{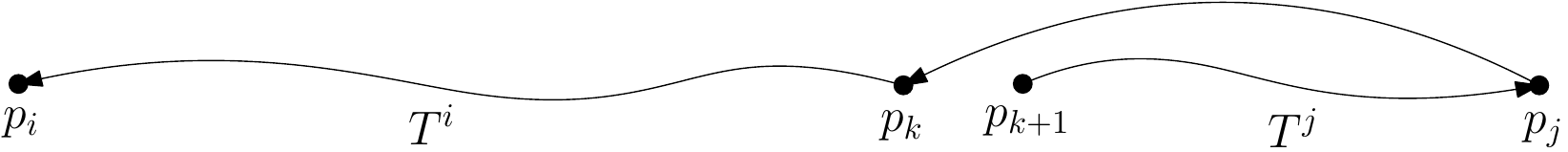}
	\caption{A \mwst $T_{[i,j]}^i$ rooted at $p_i$ . $(p_j,p_k)$ partitions $T_{[i,j]}^i$ into $T^i$ and $T^j$.} 
	\label{fig:opt}
\end{figure}
%

Based on the aforementioned, to compute $OPT_{[i,j]}^i$ (resp., $OPT_{[i,j]}^j$), we compute $OPT_{[i,k]}^i + w(p_j,p_k) + OPT_{[k+1,j]}^j$ (resp., $OPT_{[i,k-1]}^i + w(p_i,p_k) + OPT_{[k,j]}^j$), for each $i \le k <j$ (resp., for each $i < k \le j$), and take the minimum over these values. That is,

$$
OPT_{[i,j]}^i =  
\begin{cases} 
\quad \quad \quad  0 \quad \quad \quad \quad \quad \quad \quad \quad \quad \quad \quad \quad \quad \quad \quad \quad \ \ : \text{if }  i=j
\\ \underset{i \le k < j}{\min} \{OPT_{[i,k]}^i + w(p_j,p_k) + OPT_{[k+1,j]}^j\} \quad \quad \  :  \text{otherwise,}
\end{cases}
$$
and
$$
OPT_{[i,j]}^j =  
\begin{cases} 
\quad \quad \quad  0 \quad \quad \quad \quad \quad \quad \quad \quad \quad \quad \quad \quad \quad \quad \quad \quad \ \ : \text{if }  i=j
\\ \underset{i < k \le j}{\min} \{OPT_{[i,k-1]}^i + w(p_i,p_k) + OPT_{[k,j]}^j\} \quad \quad \  :  \text{otherwise.}
\end{cases}
$$

We compute $OPT_{[i,j]}^i$ and $OPT_{[i,j]}^j$, for each $1 \le i \le j \le n$, using dynamic programming as follows. 
We maintain two tables $\stackrel{\leftarrow}{S}$ and $\stackrel{\rightarrow}{S}$ each of size $n \times n$, such that $\stackrel{\leftarrow}{S}[i,j] = OPT_{[i,j]}^i$ and $\stackrel{\rightarrow}{S}[i,j] = OPT_{[i,j]}^j$, for each $1 \le i \le j \le n$. We fill the tables using Algorithm~\ref{algo:sink}.
\floatname{algorithm}{Algorithm}
\begin{algorithm}[ht]
\caption{$ComputeAllSinks(G,w)$ } \label{algo:sink}
\begin{algorithmic}[1]

\STATE \textbf{for} each $i \leftarrow 1$ to $n$ \textbf{do} \\ 
\quad  $\stackrel{\leftarrow}{S}[i,i] \leftarrow 0$ \\
\quad  $\stackrel{\rightarrow}{S}[i,i] \leftarrow 0$ \\

\STATE \textbf{for} each $d \leftarrow 1$ to $n-1$ \textbf{do} \\
\quad   \textbf{for} each $i \leftarrow 1$ to $n-d$ \textbf{do} \\
		\quad \quad \ $j \leftarrow i+d$  \\
    \quad \quad \ $\stackrel{\leftarrow}{S}[i,j] \leftarrow \underset{i \le k < j}{\min} \{\stackrel{\leftarrow}{S}[i,k] + w(p_j,p_k) + \stackrel{\rightarrow}{S}[k+1,j] \}$ \\
    \quad \quad \ $\stackrel{\rightarrow}{S}[i,j] \leftarrow \underset{1 < k \le j}{\min} \{\stackrel{\leftarrow}{S}[i,k-1] + w(p_i,p_k) + \stackrel{\rightarrow}{S}[k,j]\}$ \\
	
\end{algorithmic}
\end{algorithm} 

Notice that, when we fill the cells $\stackrel{\leftarrow}{S}[i,j]$, all the cells $\stackrel{\leftarrow}{S}[i,k]$ and $\stackrel{\rightarrow}{S}[k+1,j]$, for each $i \le k < j$, are already computed, and when we fill the cells $\stackrel{\rightarrow}{S}[i,j]$, all the cells $\stackrel{\leftarrow}{S}[i,k-1]$ and $\stackrel{\rightarrow}{S}[k,j]$, for each $i < k \le j$, are already computed. Therefore, each cell in the table is computed in $O(n)$ time, and the whole table is computed in $O(n^3)$ time.


\subsection{Solving \emph{MTIP} in 1D} \label{sec:MTIP}
In this section, we present an $O(n^3)$-time algorithm that solves \emph{MTIP} in 1D. That is, 
given a set $P = \{p_1,p_2,\dots,p_n \}$ on an horizontal line,
the algorithm computes a range assignment $\r$ for $P$, such that the graph $G_{\r}$ induced by $\r$ is strongly connected and $I(G_{\r})$ is minimized. 
For simplicity, we assume that for every $i<j$, $p_i$ is to the left of $p_j$.

Given a range assignment $\rho : P \rightarrow \mathbb{R}^+$, the interference 
of a point $p \in P$, denoted by $I_\r (p)$, is equal to 
the number of points in $P \setminus \{p\}$ of distance 
at most $\r (p)$ from $p$ (where $\r (p)$ is the range assigned to $p$), i.e., $I_\r (p)  = | \{ q \in P\setminus \{p\} : |pq| \leq \r (p) \}|$.
The cost of an assignment $\rho$, is defined as $cost(\rho) =  \sum_{p \in P} I_{\r} (p)$. 
Notice that, $cost(\r) = I(G_{\r})$, where $G_{\r}$ is the graph induced by $\r$.

Instead of assigning each point in $P$ a range,
we assign each point two directional ranges, 
\emph{left range assignment}, $\rho^l : P \rightarrow \mathbb{R}^+$,
and \emph{right range assignment}, $\rho^r : P \rightarrow \mathbb{R}^+$.
A pair of assignments $(\rho^l,\rho^r)$ is called a \emph{left-right assignment}.
Assigning a point $p \in P$ a left range $\rho^l(p)$ and a right range $\rho^r(p)$
implies that in the induced graph $G_{\rho^{lr}}$, $p$ can reach every point to its left up to distance $\rho^l(p)$
and every point to its right up to distance $\rho^r(p)$.
That is, $G_{\rho^{lr}}$  contains the directed edge $(p_i,p_j)$ if and only if one of the following holds:
(i) $i<j$ and $|p_i p_j| \leq \rho^r(p_i)$, or \ 
(ii) $j<i$ and $|p_i p_j| \leq \rho^l(p_i)$. 
A left-right assignment $(\rho^l,\rho^r)$ is called \emph{valid} if the graph induced by $(\rho^l,\rho^r)$ is strongly connected.
The cost of a left-right assignment $(\rho^l,\rho^r)$, is defined as
$cost(\rho^l,\rho^r) =  \sum_{p \in P}\max\{ I_{\r^l} (p) , I_{\r^r} (p) \} $.

Notice that each left-right assignment $(\rho^l,\rho^r)$ for $P$ can be converted to a 
range assignment $\rho$ with the same cost by assigning each point $p \in P$ a range $\rho(p)=\max\{\rho^l(p),\rho^r(p)\}$. 
On the other hand, each range assignment $\r$ for $P$ can be converted to a left-right assignment
with the same cost, by assigning each point $p \in P$, $\rho^l(p)=\rho^r(p)=\rho(p)$. To be more precise, either $\rho^l(p)$ or $\rho^r(p)$ should be reduced to $|pq|$, where $q$ is the farthest point in the appropriate direction (see Observation~\ref{obs:range_ij}).
Therefore, instead finding an optimal range assignment, our algorithm finds a left-right assignment of minimum cost. 

Given a left-right assignment $(\rho^l,\rho^r)$, let
$\stackrel{\leftarrow}{I_{\r^l}} (p_i)  = | \{ p_j \in P : j < i \text{ and } |p_ip_j| \leq \r^l (p_i) \}|$ and $\stackrel{\rightarrow}{I_{\r^r}} (p_i)  = | \{ p_j \in P : j > i \text{ and } |p_ip_j| \leq \r^r (p_i) \}|$.
In addition to the $cost$ function, we define 
$cost'(\rho^l,\rho^r) =  \sum_{p \in P}(\stackrel{\leftarrow}{I_{\r^l}} (p)+\stackrel{\rightarrow}{I_{\r^r}} (p))$,
and refine the notion of \emph{optimal solution} to include
only solutions $(\rho^l,\rho^r)$ that minimize $cost'(\rho^l,\rho^r)$
among all solutions with minimum $cost(\rho^l,\rho^r)$.

\begin{obs} \label{obs:range_ij}
Let $(\rho^l, \rho^r)$ be an optimal solution. Then, for every point $p_i \in P$, $\rho^l(p_i)=|p_j p_i|$, for some $j \leq i$, and $\rho^r(p_i)=|p_i p_k|$, for some $k \ge i$.
\end{obs}

\old{
\begin{obs} \label{obs:p_1}
In any optimal solution $(\rho^l, \rho^r)$, $\rho^l(p_1)=0$ and $\rho^r(p_1) > 0$.
\end{obs}
}

For every $1 \le i \le j \le n$, let $P_{[i,j]} \subseteq P$ be the set $\{p_i,p_{i+1},\dots ,p_j\}$.

\begin{lemma} \label{lemma:OPTprop}
There exists an optimal solution $(\rho^l, \rho^r)$ satisfying the following properties.
Let $p_i$ be a point in $P$, such that, $\rho^l(p_i)=|p_i p_j|$, for some $j < i-1$, and $\rho^r(p_i)=|p_i p_{j'}|$, for some ${j'} > i+1$. Then,
\begin{enumerate}
	\item[$(P1)$] for each point $p_k \in P_{[j+1,i-1]}$, $\r^l(p_k) < |p_kp_j|$; 
	\item[$(P2)$] for each point $p_k \in P_{[i+1,{j'}-1]}$, $\r^r(p_{k}) < |p_{k}p_{j'}|$;
	\item[$(P3)$] for each point $p_k \in P_{[j+1,i-1]}$,  $ \r^r(p_k) \le |p_kp_i|$;   and
	\item[$(P4)$] for each point $p_{k} \in P_{[i+1,{j'}-1]}$, $\r^l(p_{k}) \le |p_{k}p_i|$;
\end{enumerate}
\end{lemma}
\begin{proof} 
\textcolor{white}{zzzz}
\begin{enumerate} 
	\item[$(P1)$] Assume towards a contradiction that there exists a point $p_k \in P_{[j+1,i-1]}$, such that $\r^l(p_k) = |p_kp_l| \ge |p_kp_j|$, for some $l \le j$; see Figure~\ref{fig:noCrossing2}(a,b). 
	Let $(\r'^l, \r'^r)$ be the assignment obtained from $(\rho^l, \rho^r)$ by assigning $p_i$ a range $\r'^l(p_i) = |p_ip_k|$. Thus, (i) the graph induced by $(\r'^l, \r'^r)$ is still strongly connected, (ii) $I_{\r'^l}(p_i) < I_{\r^l}(p_i)$ and $I_{\r'^r}(p_i) = I_{\r^r}(p_i)$, and (iii) $\stackrel{\leftarrow}{I_{\r'^l}} (p_i) < \stackrel{\leftarrow}{I_{\r^l}} (p_i)$. Therefore, $(\r'^l, \r'^r)$ is a valid assignment, $cost(\r'^l, \r'^r) \le cost(\r^l,\r^r)$, and $cost'(\r'^l, \r'^r) < cost'(\r^l,\r^r)$, which contradicts the minimality of $cost'(\r^l,\r^r)$.
	\item[$(P2)$] The proof is symmetric to the proof of $(P1)$.
\begin{figure}[ht]
	\centering
			\includegraphics[width=0.93\textwidth]{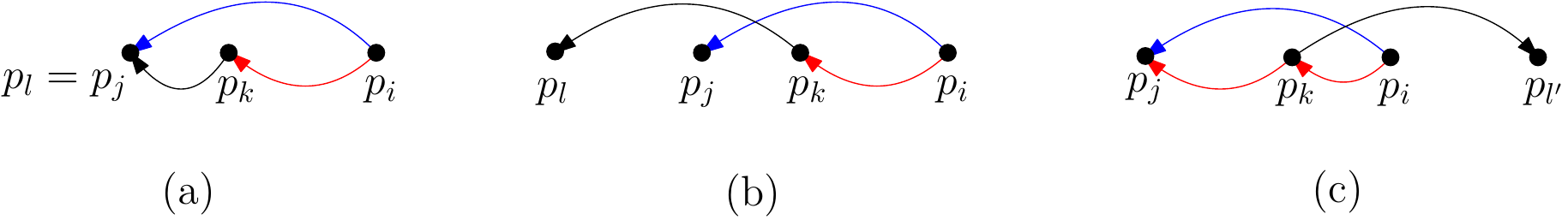}
	\caption{Illustration of  the proof of Lemma~\ref{lemma:OPTprop}.} 
	\label{fig:noCrossing2}
\end{figure}
		%
	\item[$(P3)$] Assume that there exists a point $p_k \in P_{[j+1,i-1]}$, such that  
	$\r^r(p_k) = |p_kp_{l'}| > |p_kp_i|$, for some ${l'} > i$; see Figure~\ref{fig:noCrossing2}(c). 
	Assume, w.l.o.g., that $|p_jp_k| \le |p_ip_{l'}|$. %
	By $(P1)$, $\r^l(p_k) < |p_kp_j|$. 
	Let $(\r'^l, \r'^r)$ be the assignment obtained from $(\rho^l, \rho^r)$ by assigning $p_i$ a range $\r'^l(p_i) = |p_ip_k|$ and assigning $p_k$ a range $\r'^l(p_k) = |p_kp_j|$. Thus, (i) the graph induced by $(\r'^l, \r'^r)$ is still strongly connected, (ii) $I_{\r'^l}(p_i) \le I_{\r^l}(p_i)$ and $I_{\r'^l}(p_k) \le I_{\r^r}(p_k)$, and (iii) $\stackrel{\leftarrow}{I_{\r'^l}} (p_i) + \stackrel{\leftarrow}{I_{\r'^l}} (p_k) = \stackrel{\leftarrow}{I_{\r^l}} (p_i)$. Therefore, $(\r'^l, \r'^r)$ is a valid assignment, $cost(\r'^l, \r'^r) \le cost(\r^l,\r^r)$, and $cost'(\r'^l, \r'^r) \le cost'(\r^l,\r^r)$, which implies that $(\r'^l, \r'^r)$ is an optimal solution satisfying the lemma.
	\item[$(P4)$] The proof is symmetric to the proof of $(P3)$.
\end{enumerate}
\end{proof}

Let $G=(P,E)$ be the complete directed graph on $P$, in which $w(p_i,p_j) =|\{ p_k\in P\setminus \{p_i\}: |p_ip_k| \le |p_ip_j| \}|$, for each directed edge $(p_i,p_j) \in E$.
Let $G_{[i,j]}$ be the subgraph of $G$ induced by $P_{[i,j]}$.
Let $(\r_{[i,j]}^l, \r_{[i,j]}^r)$ be an assignment for the points of $P_{[i,j]}$, such that the graph induced by $(\r_{[i,j]}^l, \r_{[i,j]}^r)$ contains a sink tree $T_{[i,j]}^x$ rooted at $p_x$, where $x \in \{i,j\}$.
In Lemma~\ref{lemma:optsink}, we proved that, for any $x \in \{i,j\}$, $(\r_{[i,j]}^l, \r_{[i,j]}^r)$ is an optimal assignment (i.e., $T_{[i,j]}^x$ is of minimum interference) if and only if $T_{[i,j]}^x$ is a \mwst rooted at $p_x$ in $G_{[i,j]}$. 
Combining this with Lemma~\ref{lemma:OPTprop}, we have the following corollary.

\begin{corollary} \label{cor:sink}
Let $(\r^l, \r^r)$ be an optimal solution satisfying the properties of Lemma~\ref{lemma:OPTprop}. Let $p_i$ be a point in $P$, such that, $\rho^l(p_i)=|p_i p_k|$, for some $k < i-1$, and $\rho^r(p_i)=|p_i p_j|$, for some $j > i+1$; see Figure~\ref{fig:sink}. Then, the graph induced by $(\r_{[k+1,i-1]}^l, \r_{[k+1,i-1]}^r)$ is a \mwst $T_{[k+1,i]}^i$ rooted at $p_i$ in $G_{[k+1,i]}$, and the graph induced by $(\r_{[i+1,j-1]}^l, \r_{[i+1,j-1]}^r)$ is a \mwst $T_{[i,j-1]}^i$ rooted at $p_i$ in $G_{[i,j-1]}$.
\end{corollary}
\begin{figure}[ht]
	\centering
			\includegraphics[width=0.52\textwidth]{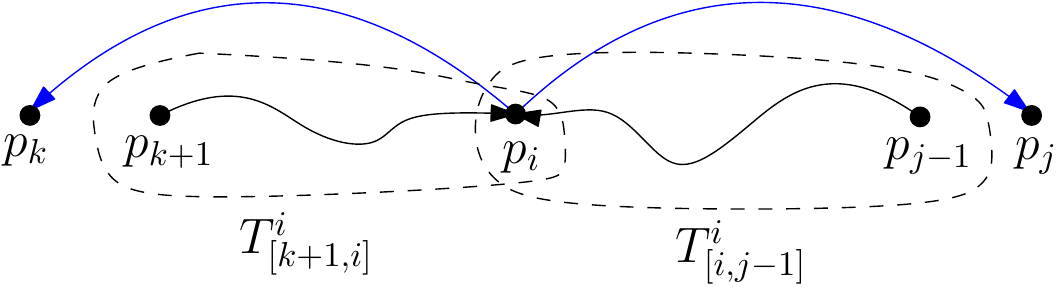}
	\caption{Illustration of Corollary~\ref{cor:sink}.} 
	\label{fig:sink}
\end{figure}

\begin{lemma} \label{lemma:P5}
Let $(\r^l, \r^r)$ be an optimal solution satisfying the properties of Lemma~\ref{lemma:OPTprop} and let $G_{\r^{lr}}$ be the strongly connected graph induced by $(\r^l,\r^r)$.
Let $p_i$ be a point in $P$, such that, $\rho^l(p_i)=|p_i p_k|$, for some $k < i$, and $\rho^r(p_i)=|p_i p_j|$, for some $j > i$. 
Then, for each point $p_{j'} \in P_{[j,n]}$, $\r^l(p_{j'}) \le |p_{j'}p_{j-1}|$.
\end{lemma}
\begin{proof}
Assume towards a contradiction that there exists a point $p_{j'} \in P_{[j,n]}$, such that 
$\r^l(p_{j'}) = |p_{j'}p_l| \ge |p_{j'}p_{j-2}|$. We distinguish between two cases. \\
\textbf{Case~1:} $l \ge i$; see Figure~\ref{fig:noCrossing3}. 
Let $p_t \in P_{[l+1,j-1]}$. By $(P1)$ in Lemma~\ref{lemma:OPTprop}, $\r^l(p_t) < |p_tp_{l}|$, and by $(P3)$ in Lemma~\ref{lemma:OPTprop}, $\r^r(p_t) < |p_tp_{j}|$. 
Thus, the points in $P_{[l+1,j-1]}$ could not be connected to the points in $P_{[1,n]} \setminus P_{[l+1,j-1]}$, which contradicts connectivity $G_{\r^{lr}}$.
%
\begin{figure}[ht]
	\centering
			\includegraphics[width=0.68\textwidth]{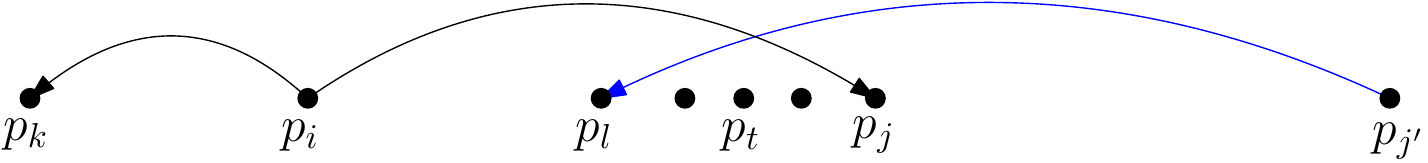}
	\caption{Illustration of Case~1 in the proof of Lemma~\ref{lemma:P5}.} 
	\label{fig:noCrossing3}
\end{figure}

\noindent
\textbf{Case~2:} $ l < i$.
By $(P1)$ in Lemma~\ref{lemma:OPTprop}, $l < k$; see Figure~\ref{fig:noCrossing4}. 
By Corollary~\ref{cor:sink}, the graph induced by $(\r_{[l+1,{j'}-1]}^l, \r_{[l+1,{j'}-1]}^r)$ is a \mwst $T_{[l+1,j']}^i$ rooted at $p_{j'}$.
Therefore, since $(\r^l, \r^r)$ is of minimum $cost'$, either $\rho^l(p_i)=0$ or $\rho^r(p_i)=0$, which contradicts the assumption that $\rho^l(p_i) > 0$ and $\rho^r(p_i) > 0$. 
\end{proof}
%
\begin{figure}[ht]
	\centering
			\includegraphics[width=0.63\textwidth]{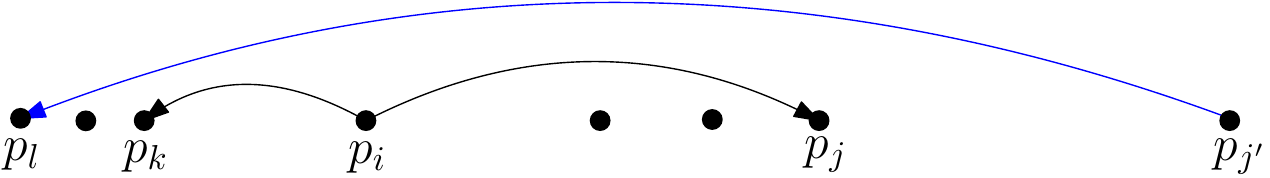}
	\caption{Illustration of Case~2 in the proof of Lemma~\ref{lemma:P5}.} 
	\label{fig:noCrossing4}
\end{figure}
%

For each $1 \le k \le i \le n$, let $OPT(i,k)$ denote the cost of an optimal solution $(\r_{[i,n]}^l, \r_{[i,n]}^r)$ for the sub-problem defined by the set $P_{[i,n]}$, in which $\r^l(p_i) = |p_ip_k|$; see Figure~\ref{fig:opt2}. Therefore, the cost of an optimal solution for the whole problem is $OPT(1,1)$. 
For each $i \le j \le n$, let $\Delta(i,j,k) = \max \{0,w(p_i,p_j)-w(p_i,p_k)\}$ denote the difference between $w(p_i,p_j)$ and $w(p_i,p_k)$. That is, 
$$
\Delta(i,j,k) = 
\begin{cases} 
 \quad w(p_i,p_j)-w(p_i,p_k)  \quad \quad \quad  : \text{if }  |p_ip_j| > |p_ip_k|
\\ \quad \quad \quad \quad 0 \quad \quad \quad \quad \quad \quad \quad \quad  :  \text{otherwise.} 
\end{cases}
$$

If $i=n$, then, clearly $OPT(i,k) = 0$. Otherwise, $\r^r(p_i) > 0$ and thus there exists a point $p_j \in P_{[i+1,n]}$, such that $\r^r(p_i) = |p_ip_j|$, and, by Lemma~\ref{lemma:OPTprop} and Lemma~\ref{lemma:P5}, there exists a point $p_t \in P_{[j,n]}$, such that $\r^l(p_t) = |p_tp_{j-1}|$. Moreover, for $i+1 < j < t$, by Corollary~\ref{cor:sink}, the graph induced by $(\r_{[i+1,j-1]}^l, \r_{[i+1,j-1]}^r)$ is a \mwst $T_{[i,j-1]}^i$ rooted at $p_i$ in $G_{[i,j-1]}$, and the graph induced by $(\r_{[j,t-1]}^l, \r_{[j,t-1]}^r)$ is a \mwst $T_{[j,t]}^t$ rooted at $p_t$ in $G_{[j,t]}$; see Figure~\ref{fig:opt2}.
If $j=i+1$ and $t=j$, then $w(T_{[i,j-1]}^i) = w(T_{[j,t]}^t) = 0$.
Therefore, \\ 
$OPT(i,k) = \Delta(i,j,k) + w(T_{[i,j-1]}^i) + w(p_t,p_{j-1}) + w(T_{[j,t]}^t) + OPT(t,j-1).$
\begin{figure}[ht]
	\centering
			\includegraphics[width=0.72\textwidth]{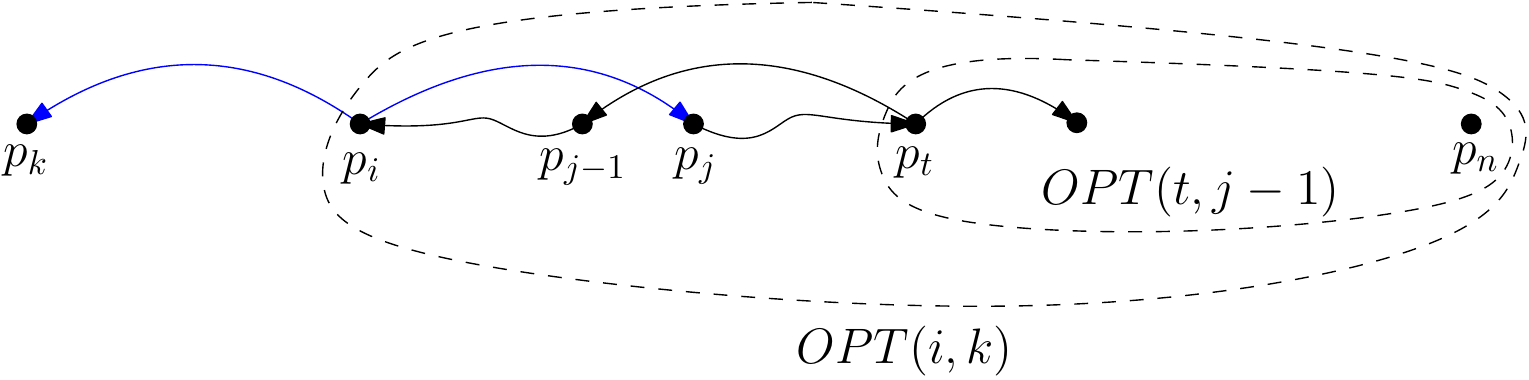}
	\caption{Computing $OPT(i,k)$.} 
	\label{fig:opt2}
\end{figure}

Based on the aforementioned, to compute $OPT(i,k)$, we compute $\Delta(i,j,k) + w(T_{[i,j-1]}^i) + w(T_{[j,t]}^t) + w(p_t,p_{j-1}) + OPT(t,j-1)$, for each $i < j \le n$ and for each $j \le t \le n$, and take the minimum over these values. That is, if $i=n$, then $OPT(i,k) = 0$, otherwise 
$$
OPT(i,k) = \underset{j \le t \le n}{\underset{i < j \le n}{\min}}
 \left\{ \Delta(i,j,k) + w(T_{[i,j-1]}^i) + w(T_{[j,t]}^t) + w(p_t,p_{j-1}) + OPT(t,j-1) \right\}.
$$
For each $i < j \le n$, let $C(j) = \underset{j \le t \le n}{\min} \left\{ w(T_{[j,t]}^t) + w(p_t,p_{j-1}) + OPT(t,j-1) \right\}$ and observe that $C(j)$ depends only on $t$.
Therefore, $$OPT(i,k) = \underset{i < j \le n}{\min}  \left\{ \Delta(i,j,k) + w(T_{[i,j-1]}^i) + C(j) \right\}.$$ 

\old{ 
$$
OPT(i,k) =  
\begin{cases} 
\quad \quad \quad \quad \quad \quad  \ 0 \quad \quad \quad \quad \quad \quad \quad \quad \quad \quad \quad \quad \quad \quad \quad \quad \quad \quad \quad \quad \quad \quad \quad \quad \ \ : \text{if }  i=n
\\ \underset{i < j \le n}{\underset{j \le t \le n}{\min}}
 \{ \Delta(i,j,k) + w(T_{[i,j-1]}^i) + w(T_{[j,t]}^t) + w(p_t,p_{j-1}) + OPT(t,j-1) \} \quad \  :  \text{otherwise.} 
\end{cases}
$$
}

We compute $OPT(i,k)$, for each $1 \le k \le i \le n$, using dynamic programming as follows. 
We maintain two tables $M$ of size $n \times n$ and $C$ of size $1 \times n$, such that $M[i,k] = OPT(i,k)$, for each $1 \le k \le i \le n$, and $C[j] = C(j)$, for each $1 \le j \le n$. In Section~\ref{sec:allSinks}, we computed two tables $\stackrel{\leftarrow}{S}$ and $\stackrel{\rightarrow}{S}$ each of size $n \times n$, such that $\stackrel{\leftarrow}{S}[i,j] = w(T_{[i,j]}^i)$ and $\stackrel{\rightarrow}{S}[i,j] = w(T_{[i,j]}^j)$, for each $1 \le i \le j \le n$. 
Algorithm~\ref{algo:MTIP} uses these tables to fill the table $M$ and returns $M[1,1] = OPT(1,1)$.  
\floatname{algorithm}{Algorithm}
\begin{algorithm}[ht]
\caption{$SolveMTIP(G,w)$ } \label{algo:MTIP}
\begin{algorithmic}[1]

\STATE \textbf{for} each $i \leftarrow n$ to $1$ \textbf{do} \\
	\quad   $C[i] \leftarrow \infty$ \\
	\quad   \textbf{for} each $k \leftarrow 1$ to $i$ \textbf{do} \\
			\quad \quad \ $M[i,k] \leftarrow \infty$  \\

\STATE \textbf{for} each $k \leftarrow 1$ to $n$ \textbf{do} \\ 
	\quad  $M[n,k] \leftarrow 0$ \\ 

\STATE \textbf{for} each $i \leftarrow n-1$ to $1$ \textbf{do} \\
	\quad  \textbf{for} each $j \leftarrow i+1$ to $n$ \textbf{do} \\
	\quad  \quad  \textbf{for} each $t \leftarrow j$ to $n$ \textbf{do} \\
	\quad  \quad  \quad \ $C[j] \leftarrow \min \{ C[j] \ , \ \stackrel{\rightarrow}{S}[j,t] + w(p_t,p_{j-1}) + M[t,j-1] \}$ \\
	\quad  \quad  \textbf{for} each $k \leftarrow 1$ to $i$ \textbf{do} \\
	\quad  \quad  \quad \ \textbf{if} $|p_ip_j| > |p_ip_k|$ \textbf{then} \\
	\quad  \quad  \quad \ \quad \ $\Delta \leftarrow w(p_i,p_j) - w(p_i,p_k)$  \\
	\quad  \quad  \quad \ \textbf{else} \\
	\quad  \quad  \quad \ \quad \ $\Delta \leftarrow 0$  \\		
  \quad  \quad  \quad  $M[i,k] \leftarrow \min \{ M[i,k] \ , \ \Delta \ + \stackrel{\leftarrow}{S}[i,j-1] + C[j] \}$ \\

\STATE \textbf{return} $M[1,1]$
    	
\end{algorithmic}
\end{algorithm}

Notice that, when we fill the cell $C[j]$, the cells $M[t,j-1]$ are already computed, for each $i < t \le n$ and for each $1 < j \le i$. Moreover, when we fill the cell $M[i,k]$, the cell $C[j]$ is already computed. 
Since for each $i< j \le n$, the cell $C[j]$ is filled $n-1$ times (for each $1 \le i < n$) and computed by taking the minimum over $\stackrel{\rightarrow}{S}[j,t] + w(p_t,p_{j-1}) + M[t,j-1]$, for each $j \le t \le n$, the total time for filling the table $C$ is $O(n^3)$.
Moreover, each cell $M[i,k]$ is computed by taking the minimum over $\Delta + \stackrel{\leftarrow}{S}[i,j-1] + C[j]$, for each $i < j \le n$. Thus, each cell $M[i,k]$ is computed in $O(n)$ time, and the whole table is computed in $O(n^3)$ time. Therefore, the total running time of Algorithm~\ref{algo:MTIP} is $O(n^3)$.

The following theorem summarizes the result of this section.
\begin{theorem} \label{thm:opt}
Let $P$ be a set of $n$ points located on a horizontal line. Then, one can compute in $O(n^3)$ time a range assignment $\r$ to the points of $P$, such that the induced graph $G_\r$ is strongly connected and its total interference is minimized.
\end{theorem}

\section{\emph{MTIP} in 2D }\label{sec:NP}

In this section, we prove that \emph{MTIP} is NP-complete in 2D and present a polynomial-time 2-approximation algorithm for the problem. 

\begin{theorem}\label{thm:NP}
MTIP in 2D is NP-complete.
\end{theorem}

\begin{proof}
Given a range assignment $\r$ as a certificate, it is easy to verify in polynomial-time whether the graph induced by $\r$ is strongly connected and whether its total interference is bounded by a given value $I$. This implies that \emph{MTIP} is in $NP$.

To prove hardness of \emph{MTIP}, we show a polynomial-time reduction from the problem of deciding whether a grid graph contains a Hamiltonian cycle, which is known to be NP-hard~\cite{Itai82}. A grid graph is a graph whose vertex set is a subset of the integer grid $\mathbb{Z} \times \mathbb{Z}$, and two vertices are connected by an edge if and only if the distance between them is equal to 1; see Figure~\ref{fig:grid}(a).

Let $G=(V,E)$ be a grid graph, where $V=\{ v_1,v_2,\dots,v_n\}$. We construct in polynomial-time a set $P$ of $5n$ points in the plane, and show that $G$ contains a Hamiltonian cycle iff there exists a range assignment $\r$ such that the induced graph $G_\r$ is strongly connected and $I(G_\r) = 9n$. We assume that the degree of each vertex in $G$ is at least 2, otherwise $G$ cannot contain a Hamiltonian cycle.
\begin{figure}[ht]
	\centering
			\includegraphics[width=0.72\textwidth]{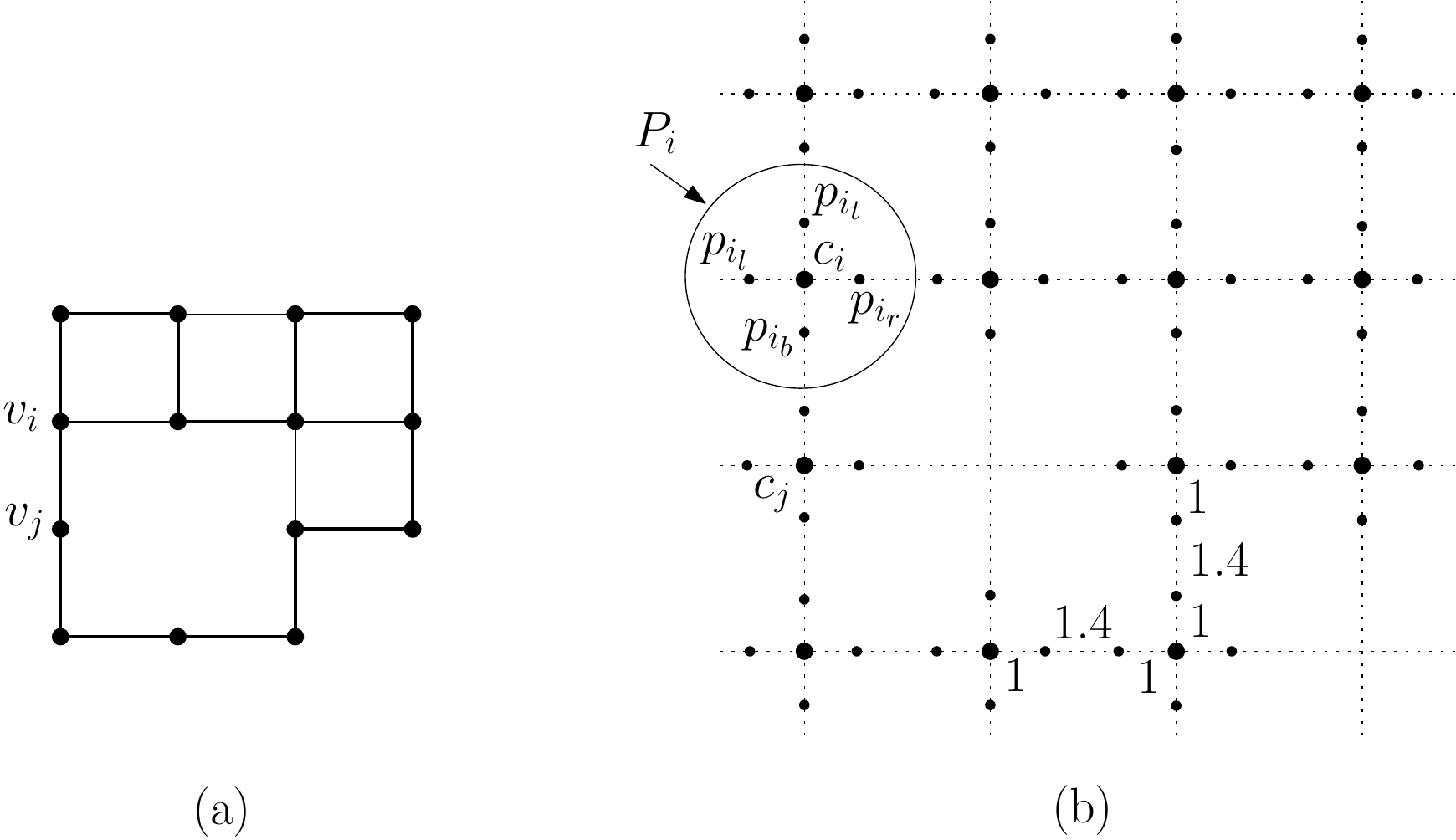}
	\caption{(a) A grid graph $G$. (The bold edges form a Hamiltonian cycle.) (b) The resulting set $P$. Each set $P_i \subseteq P$ corresponds to a vertex $v_i$ in $G$ and consists of 5 points $\{ c_i,p_{i_r}, p_{i_l}, p_{i_t},p_{i_b} \}$.} 
	\label{fig:grid}
\end{figure}

We first transform the vertices of $G$ to a set  $C=\{ c_1,c_2,\dots,c_n\}$ of $n$ points on a grid of side length $3.4$, such that two vertices $v_i$ and $v_j$ are adjacent in $G$ if and only if $c_i$ and $c_j$ (the points corresponding to $v_i$ and $v_j$) are adjacent in the new grid; see Figure~\ref{fig:grid}.
Then, for each point $c_i \in P$, we locate four points $p_{i_r}, p_{i_l}, p_{i_t},$ and $p_{i_b}$ on the grid edges incident to $c_i$, such that the distance between $c_i$ and each one of them is equal to 1; see Figure~\ref{fig:grid}(b). Put $P_i= \{ c_i,p_{i_r}, p_{i_l}, p_{i_t},p_{i_b} \}$. We will refer to $c_i$ as the center of $P_i$ and to the other four points as connectors. Let $P = \underset{v_i \in V}{\bigcup} P_i$ be the resulting set. Clearly, $|P| = 5n$ and $P$ can be constructed in polynomial-time.

%
\begin{lemma}\label{lemma:int9}
Let $\r$ be a valid range assignment of $P$ and let $G_\r$ be the graph induced by $\r$. Then, for each $1 \le i \le n$, $SI(P_i) = \sum_{p \in P_i}SI(p) \ge 9$. 
\end{lemma}
\begin{proof}
Since $\r$ is a valid assignment, for each $p\in P_i$, we have $\r(p) \ge 1$. Hence, $SI(c_i) \ge 4$ and $SI(p) \ge 1$, for each $p \in P_i \setminus \{c_i\}$. Moreover, since $G_\r$ is strongly connected, the transmission range of at least one of the points of $P_i$ must cover at least one point from $P \setminus P_i$, which means that either $\r(c_i) \ge 2.4$ or $\r(p) \ge 1.4$ for at least one connector $p \in P_i$. 
If $\r(c_i) \ge 2.4$, then $SI(c_i) \ge 6$ (since the degree of each vertex in $G$ is at least 2), and therefore, $SI(P_i) \ge 10$; see Figure~\ref{fig:int9}(a). Otherwise, at least one connector $p \in P_i$ has $\r(p) \ge 1.4$. Then, $SI(p) \ge 2$, and therefore $SI(P_i) \ge 9$; see Figure~\ref{fig:int9}(b). 
\end{proof}
\begin{figure}[ht]
	\centering
			\includegraphics[width=0.8\textwidth]{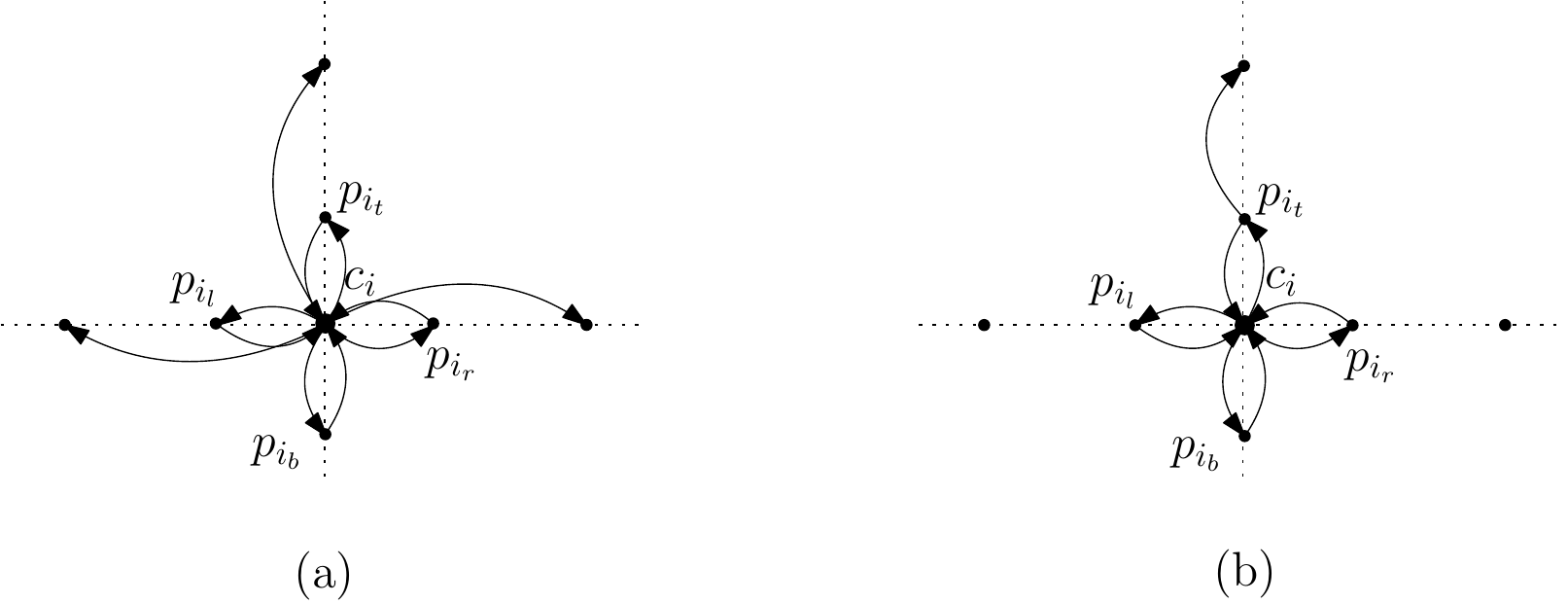}
	\caption{$SI(P_i) \ge 9$.} 
	\label{fig:int9}
\end{figure}

\begin{corollary}\label{cor:only1long}
Let $\r$ be a valid range assignment of $P$. Then, if $SI(P_i) = 9$, for some $1 \le i \le n$, then $\r(c_i) < 2.4$, exactly one connector $p \in P_i$ has $1.4 \le \r(p) < \sqrt{2}$, and each of the other three connectors $p' \in P_i$ has $1 \le \r(p') < 1.4$.
\end{corollary}
\begin{proof}
During the proof of Lemma~\ref{lemma:int9}, we showed that if $\r(c_i) \ge 2.4$, then $SI(c_i) \ge 6$, and hence $SI(P_i) \ge 10$. Thus, $\r(c_i) < 2.4$ and $SI(c_i) = 4$. Moreover, for each point $p \in P_i\setminus \{c_i\}$, $\r(p) \ge 1$, and if $1.4 \le \r(p) < \sqrt{2}$, then $SI(p) = 2$. Since $SI(P_i) = 9$, exactly one point $p \in P_i\setminus \{c_i\}$ has $1.4 \le \r(p) < \sqrt{2}$, which completes the proof of the lemma.
\end{proof}

We now prove the correctness of the reduction. Suppose that $G$ contains a Hamiltonian cycle $C$. We compute a valid range assignment $\r$ to the points of $P$, such that $I(G_\r) =9n$. Consider $C$ as a directed cycle, such that, each vertex in $C$ has in-degree 1 and out-degree 1. For each vertex $v_i$ in $G$ we assign ranges to the points of $P_i$ as follows. We assign 1 to the center $c_i$, assign $1.4$ to one of the connectors (according to the outgoing edge incident to $v_i$ in $C$), and assign 1 to each of the other three connectors. Since $C$ is a Hamiltonian cycle, the graph induced by $\r$ is strongly connected. Moreover, $SI(P_i) = 9$, for each $1 \le i \le n$, and therefore $I(G_\r) =9n$.

Conversely, suppose that there exists a valid range assignment $\r$ to the points of $P$, such that $I(G_\r) =9n$.
By Lemma~\ref{lemma:int9}, $SI(P_i) \ge 9$, for each $1 \le i \le n$. Since $I(G_\r) =9n$, we conclude that $SI(P_i) = 9$, for each $1 \le i \le n$. Moreover, by Corollary~\ref{cor:only1long}, in each set $P_i$, exactly one connector $p$ has $1.4 \le \r(p) < \sqrt{2}$ (where $\r(c_i) < 2.4$ and each of the other connectors $p'$ has $1 \le \r(p') < 1.4$). 
We construct a Hamiltonian cycle $C$ in $G$ as follows. For every two sets $P_i$ and $P_j$, we add the edge $(v_i,v_j)$ to $C$ if and only if the connector $p \in P_i$ assigned range $1.4 \le \r(p) < \sqrt{2}$ covers a point in $P_j$. Thus, $C$ is a subgraph of $G$, and, since $G_\r$ is strongly connected, $C$ is connected. Moreover, since each set $P_i$ has exactly one point which reaches a point not in $P_i$, the degree of each vertex in $C$ is exactly $2$. Therefore, $C$ is a Hamiltonian cycle in $G$.   
\end{proof}

\subsection{Approximation Algorithm for \emph{MTIP} in 2D} \label{sec:approx}


Let $P$ be a set of $n$ points in the plane. Let $\r^*$ be an optimal range assignment to the points of $P$. Let $G_{\r^*}$ be the graph induced by $\r^*$, and put $OPT = I(G_{\r^*})$.
In this section, we present a polynomial-time approximation algorithm that computes a valid range assignment $\r$, such that $I(G_\r) \le 2 \cdot OPT$.

Let $s$ be a point of $P$. 
A \emph{broadcast} tree for $P$ rooted at $s$ is a directed tree rooted at $s$, which contains a directed path from $s$ to each point $p \in P\setminus \{s\}$. 
A \emph{sink} tree for $P$ rooted at $s$ is a directed tree rooted at $s$, which contains a directed path from each point $p \in P\setminus \{s\}$ to $s$. 
We introduce two variants of \emph{MTIP}, namely \emph{MTIP}$_1$ and \emph{MTIP}$_2$. 
In \emph{MTIP}$_1$, the goal is to compute a range assignment $\r_1$ to the points of $P$, such that the graph induced by $\r_1$ contains a broadcast tree $T_{\r_1}$ for $P$ rooted at $s$ of minimum total interference.
And, in \emph{MTIP}$_2$, the goal is to compute a range assignment $\r_2$ to the points of $P$, such that the graph induced by $\r_2$ contains a sink tree $T_{\r_2}$ for $P$ rooted at $s$ of minimum total interference.

Let $\r_1$ and $\r_2$ be optimal range assignments for \emph{MTIP}$_1$ and \emph{MTIP}$_2$, respectively, and let $T_{\r_1}$ and $T_{\r_2}$ be the corresponding broadcast and sink trees.
We compute a new range assignment $\r$ as follows. For each point $p \in P$, we set $\r(p) = \max\{ \r_1(p),\r_2(p)\}$.
Let $G_\r$ be the graph induced by $\r$. Then, $G_\r$ is strongly connected, since given two points $p,q \in P$, one can get from $p$ to $q$ by first following the directed path in $T_{\r_2}$ from $p$ to $s$ and then following the directed path in $T_{\r_1}$ from $s$ to $q$. In the next lemma we bound $I(G_\r)$.

\begin{lemma} \label{lemma:2opt}
$I(G_\r) \le 2 \cdot OPT$.
\end{lemma}
\begin{proof}
Consider the graph $G_{\r^*}$. Since $G_{\r^*}$ is strongly connected, there exists a directed path from $s$ to each of the points in $P \setminus \{c\}$ and vice versa. Thus, $G_{\r^*}$ contains broadcast and sink trees rooted at $s$. Let $T_1$ and $T_2$ be such broadcast and sink trees, respectively. Clearly, $I(T_1) \le OPT$ and $I(T_2) \le OPT$. Since $T_{\r_1}$ is a broadcast tree of minimum total interference and $T_{\r_2}$ is a sink tree of minimum total interference, we have $I(T_{\r_1}) \le I(T_1)$ and $I(T_{\r_2}) \le I(T_2)$. Moreover, since $I(G_\r) \le I(T_{\r_1}) + I(T_{\r_2})$, we have $I(G_\r) \le 2 \cdot OPT$.
\end{proof}

We now show how to solve \emph{MTIP}$_1$ and \emph{MTIP}$_2$ optimally, and therefore, by Lemma~\ref{lemma:2opt}, we can obtain a valid range assignment $\r$, such that $I(G_\r) \le 2 \cdot OPT$.

\subsubsection*{Solving \emph{MTIP}$_1$}
An optimal range assignment $\r_1$ for \emph{MTIP}$_1$ can be found easily. We assign $s$ the range $\r_1(s)=\max_{p \in P}|sp|$, and for each point $p \in P \setminus \{s\}$, we assign $p$ the range $\r_1(p) = 0$. Clearly, $\r_1$ induces a broadcast tree $T_{\r_1}$ rooted at $s$ and $I(T_{\r_1}) = n-1$, which is optimal.  

\subsubsection*{Solving \emph{MTIP}$_2$}
The algorithm in this case is more involved. 
Let $G = (P,E)$ be the complete weighted directed graph on $P$, such that the weight of $(p,q) \in E$ is $w(p,q)=|\{ z\in P\setminus \{p\}: |pz| \le |pq| \}|$.
Let $T = (P,E_T)$ be a sink tree rooted at $s$ in $G$, and let $w(T) = \sum_{(p,q) \in E_T} w(p,q)$ denote its weight.  
Let $\r_2$ be a range assignment to the points of $P$, such that the induced graph contains a sink tree $T_{\r_2}$ rooted at $s$. 
Observe that if $\r_2$ is an optimal range assignment for \emph{MTIP}$_2$, then the induced graph contains a unique sink tree $T_{\r_2}$ rooted at $s$. The following lemma generalizes Lemma~\ref{lemma:optsink} and its proof is identical.
\begin{lemma} \label{lemma:optMTIP2}
$\r_2$ is an optimal range assignment for MTIP$_2$ if and only if $T_{\r_2}$ is a minimum-weight sink tree rooted at $s$ in $G$.
\end{lemma}

By Lemma~\ref{lemma:optMTIP2}, to compute an optimal range assignment $\r_2$ for \emph{MTIP}$_2$, it is sufficient to compute a minimum-weight sink tree rooted at $s$ in $G$.
We compute a minimum-weight sink tree rooted at $s$ in $G$ using Edmonds' algorithm~\cite{Edmonds67} (for finding minimum directed spanning trees in directed graphs). 

We construct the (inverse) complete weighted directed graph $G' = (P,E')$, such that the weight of $(p,q) \in E'$ is $w'(p,q) = w(q,p) = |\{ z\in P\setminus \{q\}: |qz| \le |qp| \}|$. 
Clearly, $T'$ is a broadcast tree rooted at $s$ in $G'$ of weight $W$ if and only if $T$ (the tree obtained by inverting the edges of $T'$) is a sink tree rooted at $s$ in $G$ of weight $W$.
Therefore, in order to compute a minimum-weight sink tree rooted at $s$ in $G$, it suffices to compute a minimum-weight broadcast tree rooted at $s$ in $G'$.
Since a minimum-weight broadcast tree rooted at $s$ in $G'$ can be computed in $O(n^2)$ time~\cite{Tarjan77} using Edmonds' algorithm, we can compute a minimum-weight sink tree rooted at $s$ in $G$ in $O(n^2)$ time. Moreover, since $G$ and $G'$ can be constructed in $O(n^2)$ time, we can solve \emph{MTIP}$_2$ in $O(n^2)$ time (applying Lemma~\ref{lemma:optMTIP2}).

The following theorem summarizes the result of this section.
\begin{theorem} \label{thm:2approx}
Let $P$ be a set of $n$ points in the plane. Then, one can compute in $O(n^2)$ time a valid range assignment $\r$ to the points of $P$, such that the graph $G_\r$ induced by $\r$ is strongly connected and its total interference is at most $2 \cdot OPT$.
\end{theorem}


\bibliographystyle{plain}

\end{document}